\documentclass{article}
\usepackage{lorenzo}
\usepackage{graphicx} 
\usepackage[backend=bibtex,maxbibnames=99,minbibnames=99]{biblatex}
\usepackage[utf8]{inputenc}
\title{A Tight Analysis of Khatri-Rao Oblivious Subspace Embeddings}

\author[1]{Lorenzo Beretta}
\author[2]{Cameron Musco}
\affil[1]{IBM, \texttt{lorenzo2beretta@gmail.com}}
\affil[2]{UMass Amherst, \
\texttt{cameron.musco@gmail.com}}
\date{}

\begin{document}

\vspace{-2em} \maketitle

\begin{abstract}
    We study random sketching matrices with Khatri-Rao structure. In particular, we consider the Khatri-Rao product (i.e., column-wise tensor product) $A_1\odot\cdots\odot A_d \in \R^{(n_1 \cdots n_d) \times m}$ of random matrices $A_i \in \R^{n_i \times m}$ whose columns are isotropic, independent and sub-Gaussian (e.g., Gaussian matrices).  Khatri-Rao sketching matrices are widely applied in randomized algorithms for linear algebraic computation and data analysis, when the input data has tensor structure that allows for fast multiplication with $A_1\odot\cdots\odot A_d$. However, existing theory is not able to fully explain their performance in practice. In particular, despite significant attention, our best bounds for the important \emph{oblivious subspace embedding}  property with Khatri-Rao matrices lag behind what is achievable with standard  unstructured matrices.

   For embedding a $k$-dimensional subspace to $(1\pm \epsilon)$ error, Bujanović et al. \cite{bujanovic2025subspace} prove that sketching dimension $m = O(k^{3/2}/\epsilon^2)$ suffices in the special case of $d = 2$. Their dependence on $k$ is weaker than the tight bound of $O(k/\epsilon^2)$ known for unstructured sub-Gaussian sketching matrices. In this work, we close this gap, showing that $m = \tilde O(k/\epsilon^2)$ suffices for subspace embedding with a Khatri-Rao sketching matrix with any fixed order $d$. Our proof is simple, leveraging just two basic properties of the Khatri-Rao sketching distribution: 1) the columns of $A_1\odot\cdots\odot A_d \in \R^{(n_1 \cdots n_d) \times m}$ are independent and isotropic, and 2) each column of $A_1\odot\cdots\odot A_d \in \R^{(n_1 \cdots n_d) \times m}$ satisfies a weak Johnson-Lindenstrauss type moment property.
\end{abstract}

\section{Introduction}

    Random sketching is a key tool in work on fast algorithms for linear algebraic computation, large-scale data analysis, and low-memory streaming computation \cite{muthukrishnan2005data,mahoney2011randomized,woodruff2014sketching}. Given an input matrix $X \in \R^{n \times p}$, we generate a \emph{random sketching matrix} $\Psi \in \R^{n \times m}$, and compute the sketch $\Psi^T X \in \R^{m \times p}$. $\Psi^T X$ serves as a lossy compression of $X$, that can be used to approximately solve a breadth of problems, such as  low-rank approximation \cite{halko2011finding}, linear regression \cite{sarlos2006improved}, clustering \cite{cohen2015dimensionality}, eigenvalue problems \cite{swartworth2023optimal}, and beyond.

    In many applications, $X$ is known to be structured -- e.g., it may be sparse or may be given as the tensor product of several small matrices. In such applications, it is desirable to use a sketching matrix $\Psi$ that also has structure, that facilitates fast computation of $\Psi^T X$.
In applications where $X$ has tensor structure, a key tool in the literature is the \emph{Khatri-Rao random sketch} \cite{pham2013fast,avron2014subspace,ahle2020oblivious,sun2021tensor,chen2021tensor,bamberger2022johnson,bujanovic2025subspace,saibaba2025improved,camano2025faster,haselby2026fast}.

    \begin{boxdefinition}[Khatri-Rao product]
Let $A\in\bbR^{n_a\times m}, B \in\bbR^{n_b\times m}$ have columns
$A = [a_{1},\dots,a_{m}]$ and $B = [b_{1},\dots,b_{m}]$. Their \emph{Khatri-Rao product}
is the column-wise Kronecker product
\[
A\odot B \coloneq
\begin{bmatrix}
a_{1}\otimes b_{1},
& \cdots, &
a_{m}\otimes b_{m}
\end{bmatrix}
\in\bbR^{(n_a\cdot n_b)\times m}.
\]
\end{boxdefinition}

In this work we consider a sketching matrix $\Psi = \frac{1}{\sqrt m} A_1\odot\cdots\odot A_d$ that is the Khatri-Rao product of sketching matrices $A_i \in \R^{n_i \times m}$, appropriately scaled. We will focus on the common case when all $A_i$ have independent sub-Gaussian columns, which of course includes independent standard Gaussian matrices. Although another important variant is the \emph{Tensor Sketch} of Pham and Pagh \cite{pham2013fast}, where the $A_i$ are sparse Count-Sketch matrices.
    
When the columns of $X \in \R^{(n_1 \cdots n_d) \times d}$ are also $d$-order tensor products, one can compute $\Psi^T X$ in $O((n_1 + \cdots + n_d) \cdot dm)$ time, as compared to $O(n_1 \cdots n_d) \cdot dm)$ time for a dense unstructured random projection. This represents an exponential savings in cost. For this reason, Khatri-Rao sketches have found applications in numerous settings, ranging from contour integral-based eigensolvers \cite{bujanovic2025subspace}, to compression algorithms for massive tensors and tensor networks, arising e.g., in quantum physics \cite{camano2026successive,haselby2026fast}, to approximation algorithms for polynomial kernels \cite{avron2014subspace,ahle2020oblivious}, and more \cite{kressner2017recompression,camano2025faster,chen2021tensor,bujanovic2021norm}

\subsection{Oblivious Subspace Embeddings from Khatri-Rao Sketches}

Given their importance, significant work has studied theoretical error bounds for Khatri-Rao random sketches. Much of this work has focused on two related guarantees: $(\epsilon,\delta)$-\emph{Johnson-Lindenstrauss embedding} and $(\epsilon,\delta,k)$-\emph{oblivious subspace embedding (OSE)}. For the first, it has been shown  by Ahle et al. \cite{ahle2020oblivious} that for any $\epsilon,\delta \in (0,1)$ and any fixed $x \in \R^{n_1 \cdots n_d}$, $(1-\epsilon) \norm{x}_2 \le \norm{\Psi^T x}_2 \le (1+\epsilon) \norm{x}_2$ with probability at least $1-\delta$ when the sketching dimension is at least 
\begin{align}\label{eq:jl}
    m = O_d \left (\frac{\log(1/\delta)}{\epsilon^2} + \frac{\log(1/\delta)^d}{\epsilon} \right ),
\end{align}
where $O_d(\cdot)$ hides $2^{O(d)}$ factors.
Similar bounds appear e.g., in \cite{sun2021tensor,jin2021faster,bamberger2022johnson,haselby2026fast}. As compared to dense unstructured random projections, which achieve the same bound with $m = O(\log(1/\delta)/\epsilon^2)$, the main loss in using a Khatri-Rao structured sketch is the dependence on $\log(1/\delta)^d$.
This work focuses on oblivious subspace embeddings, defined below.

\begin{boxdefinition}[Oblivious subspace embedding]
Let $\eps, \delta \in (0, 1)$ and let $k \le N$ be positive integers. A random matrix $\Psi \in \bbR^{N \times m}$ is an \emph{$(\eps, \delta, k)$-oblivious subspace embedding} (OSE) if, for every $k$-dimensional subspace $V \subseteq \bbR^{N}$,
\[
\Pru{\Psi}{
(1-\eps)\norm{x}_2 \le \norm{\Psi^\top x}_2 \le (1+\eps)\norm{x}_2
\ \text{ for all } x \in V
} \ge 1 - \delta.
\]
\end{boxdefinition}

The OSE property is a key tool in deriving theoretical analysis for sketched linear regression, low-rank approximation, and many other linear algebraic problems \cite{sarlos2006improved,woodruff2014sketching}. It can be derived from a Johnson-Lindenstrauss bound via a standard net argument -- it suffices for $\Psi$ to simultaneously preserve the norm of $2^{O(k)}$ points that form a covering of the unit sphere in the subspace. To ensure this with probability at least $1-\delta$, one sets the failure probability for embedding to $\delta/2^{O(k)}$ and applies a union bound. Unfortunately, following this approach and applying the bound of \eqref{eq:jl} gives embedding dimension for Khatri-Rao $(\epsilon,\delta,k)$-OSEs of:
\begin{align}\label{eq:netBound}
m = O_d \left (\frac{\log(2^{O(k)}/\delta)}{\epsilon^2} + \frac{\log(2^{O(k)}/\delta)^d}{\epsilon} \right ) = O_d \left (\frac{k \log(1/\delta)}{\epsilon^2} + \frac{(k+ \log(1/\delta))^d}{\epsilon} \right ).
\end{align}
See e.g., \cite{saibaba2025improved,haselby2026fast} for examples of works that derive and apply such a bound.
As compared to standard unstructured embeddings, which achieve $O(k \log(1/\delta)/\epsilon^2)$, the bound has a significantly suboptimal dependence on $k$, even e.g., when $d = 2$, where the bound is quadratic, rather than linear in $k$.

Recently, Bujanovi\'c, Grubi\v{s}i\'c, Kressner, and
Lam~\cite{bujanovic2025subspace} partially resolved this gap, focusing on the special case of $d = 2$. In this setting, they give an improved bound of 
\begin{align}\label{eq:kressner}
m = O \left (
\frac{k^{3/2} + k\log(1/\delta)}{\epsilon^2} + \frac{k^{1/2}\log(1/\delta)^2}{\epsilon} \right ).
\end{align}
However, this bound still loses a $\sqrt{k}$ factor as compared to standard unstructured sketches. Cama{\~n}o, Epperly, Meyer, and Tropp \cite{camano2025faster} show that the gap can be closed if one instead asks for a weaker \emph{oblivious subspace injection} property, which is sufficient e.g., in the analysis of some low-rank approximation algorithms. However, deriving tight bounds for oblivious subspace embedding with Khatri-Rao sketches, or showing that existing gaps from unstructured sketches are inherent, has remained open.

\subsection{Our Results}

We resolve this question, showing that Khatri-Rao random sketches can indeed nearly match unstructured sketches by achieving OSEs with a nearly linear dependence on the subspace dimension $k$. 
\begin{restatable}[Main Theorem]{theorem}{KhatriRaoMain}
\label{thm: main thm}
Let $\eps, \delta \in (0, 1)$, $n_1, \dots , n_d \geq 2$ and $k \le N = \prod_{i=1}^{d} n_i$. 
Let $A_i \in \bbR^{n_i \times m}$ be independent random matrices whose columns are independent, isotropic and have constant sub-Gaussian norm. Let $\Psi = \frac{1}{\sqrt{m}} A_1 \odot \cdots \odot A_d \in \bbR^{N \times m}$. 
Then $\Psi$ is an
$(\eps,\, \delta,\,k)$-OSE provided that
\[
m \ge
2^{O(d)} \cdot \frac{k}{\eps^2} \cdot \left(\log \left(\frac{k}{\eps \delta}\right)\right)^{d+1}.
\]
\end{restatable}

Fixing $\epsilon,\delta$ to be constants, in the case of $d = 2$, our bound improves on the existing state-of-the-art bound in \eqref{eq:kressner} by replacing a $\sqrt{k}$ with a $\log(k)^3$. For $d > 2$, the best known prior bound is \eqref{eq:netBound}, which, we improve by replacing a $k^d$ with a $k \log(k)^{d+1}$.

\smallskip

\noindent\textbf{Proof Overview.} 
Let $\Phi = A_1 \odot \dots \odot A_d$ be the unnormalized sketch, so $\Psi = \frac{1}{\sqrt m} \Phi$.
The proof of our main result is elementary and uses just two simple facts about Khatri-Rao random sketches:
\begin{enumerate}
    \item \textbf{The columns of $\Phi \in \R^{N \times m} $ are independent and isotropic.} When the columns of $A_1, \dots, A_d$ are independent and isotropic, this is straightforward.
    Indeed, the columns of $\Phi$ are independent since they are tensor products of different independent columns of $A_1,\ldots,A_d$. They are isotropic since they are tensorizations of the independent isotropic columns of $A_1,\ldots,A_d$.
    \item \textbf{The columns of $\Phi \in \R^{N \times m}$ satisfy a weak Johnson-Lindenstrauss type moment property.} In particular, for any $U \in \R^{N \times k}$ with orthonormal columns, we show that, with probability at least $1-\delta$,
    \[
    \norm{U^T \Phi_{*,i}}_2^2 \le 2^{O(d)} \cdot \norm{U}_F^2 \log(1/\delta)^d = 2^{O(d)} \cdot k \log(1/\delta)^d.
    \] 
     This bound is similar in spirit and proven using similar tools to existing Johnson-Lindenstrauss bounds for Khatri-Rao sketches, like \eqref{eq:jl}. In fact, a very similar bound is shown in \cite{bamberger2022hanson}, although they give bounds that are weaker in our parameter regime.
\end{enumerate}

With these two facts in hand, we prove our result using a standard matrix Chernoff bound \cite{tropp2012user}. In particular, $\Psi$ is an $(\epsilon,\delta,k)$-OSE if for any orthonormal $U \in \R^{N \times k}$, with probability at least $1-\delta$, for all unit-norm $x \in \R^k$, 
$$
(1-\epsilon) \le x^T U^T \Psi \Psi^T U x \le (1+\epsilon).\quad\text{Equivalently if:}\quad \norm{U^T \Psi \Psi^T U - I}_{\op} \le \epsilon.
$$

By the observation that the columns of $\Phi$ are isotropic, we have $\E[U^T \Psi \Psi^T U] = I$.  So our bound amounts to showing that $\norm{U^T \Psi \Psi^T U - \E[U^T \Psi \Psi^T U ]}_{\op} \le \epsilon$. Writing $U^T \Psi \Psi^T U$ as an average of $m$ independent outer products $$U^T \Psi \Psi^T U= \frac 1m \sum_{i=1}^m U^T \Phi_{*,i} \Phi_{*,i}^T U,$$ this deviation is exactly what a matrix Chernoff bound allows us to bound. We just need an upper bound $R^2$ on the squared norms $\norm{U^T \Phi_{*,1}}_2^2,\ldots, \norm{U^T \Phi_{*,m}}_2^2$. But this is given (with high enough probability) by our weak Johnson-Lindenstrauss type bound. In particular, union bounding over all $m$ columns, we know that, with high probability, all squared norms are simultaneously bounded by $R^2 = O(k \log(m/\delta)^{d})$. We can simply \emph{truncate} the vectors in the low-probability event that the bound fails, giving a small error which is absorbed into the final OSE bound.

Ultimately, just applying the standard matrix Chernoff bound, this gives that an embedding dimension of $m$ suffices as long as 
$$
m \gtrsim \frac{R^2 \log(k/\delta)}{\epsilon^2} =  \frac{k \log(m / \delta)^d \log(k/\delta)}{\epsilon^2},
$$ 
which holds for $m = O_d \left ( \eps^{-2} \cdot k \log(k/(\epsilon \delta))^{d+1} \right ) $. This gives the bound stated in \Cref{thm: main thm}.


\section{Proof of the Main Theorem}

Throughout, $\Phi = A_1 \odot \cdots \odot A_d \in \bbR^{N \times m}$ and $\Psi = \frac{1}{\sqrt{m}} \Phi$
is the Khatri--Rao sketch of \Cref{thm: main thm}, with $N = \prod_{i=1}^d n_i$. 
We fix an arbitrary $k$-dimensional subspace and denote by
$U \in \bbR^{N \times k}$ a matrix whose columns form an orthonormal basis of that subspace.

In \Cref{sec:iso columns} we prove that $\Phi_{*, 1} \dots \Phi_{*, m}$ are
\emph{isotropic} and statistically independent. In \Cref{sec:norm concentrates} we prove a \emph{tail bound} bounding the probability that $\|U^\top \Phi_{*, i}\|_2^2$ is much larger than its expectation.
In \Cref{sec:final proof}, we leverage these two properties along with the matrix Chernoff bound to prove \Cref{thm: main thm}.

\subsection{The columns of $\Phi$ are isotropic and independent}
\label{sec:iso columns}
The following lemma states that a Kronecker product of independent isotropic vectors
is again isotropic.

\begin{boxlemma}
\label{lem: isotropy}
Let $g_i \in \bbR^{n_i}$ for all $i \in [d]$ be statistically independent, isotropic, random vectors. Then, $\phi = g_1 \otimes \cdots \otimes g_d$ is isotropic,
that is, $\Ex{\phi \phi^\top} = I$.
\end{boxlemma}

\begin{proof}
It suffices to treat two factors, since the general case then follows by
induction on $d$.  Let $g \in \bbR^{n_1}$ and $h \in \bbR^{n_2}$ be independent
and isotropic.  Their Kronecker product $g \otimes h$ has entries
$(g \otimes h)_{(i,j)} = g_i h_j$, so
\[
\Ex{(g \otimes h)_{(i,j)}\,(g \otimes h)_{(i',j')}}
=
\Ex{g_i g_{i'}}\,\Ex{h_j h_{j'}}
=
\ind\{i = i'\}\,\ind\{j = j'\},
\]
using independence for the factorization and isotropy of $g$ and $h$ for each
factor.  Thus, $\Ex{(g \otimes h)(g \otimes h)^\top} = I_{n_1 n_2}$.
Applying this identity repeatedly to $g_1 \otimes \cdots \otimes g_d$ shows that
$\phi$ is isotropic.
\end{proof}

\begin{boxobservation}
The columns of $\Phi$ are statistically independent as they are defined as Kronecker products of statistically independent vectors.
\end{boxobservation}

\subsection{Upper tail bound for $\|U^\top \Phi_{*, i}\|_2^2$}
\label{sec:norm concentrates}
Isotropy of $\Phi_{*, i}$ implies that $\Ex{\|U^\top \Phi_{*, i}\|_2^2} = k$.
The following lemma states that $\norm{U^\top \Phi_{*, i}}_2^2$ exceeds $k$ only by a factor of $O((1+d^{-1}\log(1/\delta))^d)$ at failure probability $\delta$.

\begin{boxproposition}
\label{lem: tail}
Let $g_i \in \bbR^{n_i}$ for all $i \in [d]$ be independent, isotropic, sub-Gaussian random vectors with $\max_{i \in [d]}\|g_i\|_{\psi_2} = O(1)$. Let $\phi = g_1 \otimes \cdots \otimes g_d \in \bbR^N$.  Then
there is a constant $C = 2^{O(d)}$ such that for every $t \ge 1$ and every
$U \in \bbR^{N \times k}$ with orthonormal columns,
\[
\Pr{\norm{U^\top \phi}_2^2 > C\,k \, t}
\le
\exp\left(-d \cdot t^{1/d}\right).
\]
\end{boxproposition}

Writing $U = [U_1,\dots,U_k]$ and $y \coloneq U^\top \phi \in \bbR^k$, each
coordinate $y_j = \inner{U_j}{\phi}$ is a degree-$d$ polynomial in the entries of
the independent factors $g_1,\dots,g_d$.  Concretely, reshaping the vector $U_j$ into an order-$d$ tensor
$T = T^{(j)}$ with $T_{i_1 \cdots i_d} = (U_j)_{(i_1,\dots,i_d)}$, we may write
\[
y_j
=
\sum_{i_1,\dots,i_d}
T_{i_1 \cdots i_d}\,(g_1)_{i_1} \cdots (g_d)_{i_d}.
\]
We would like to bound the moments of $y_j$.  For $d=1$ this is
standard: a linear combination of coordinates of a vector with constant sub-Gaussian norm has
moments growing like $\sqrt p$, which we quote in the form we will use.

\begin{boximportedtheorem}[{\cite[Prop.~2.5.2]{vershynin2018high}}]
\label{thm: subgaussian moment}
Let $h = (h_1,\dots,h_n)$ be a random vector with constant sub-Gaussian norm. Then, $h$ has constant sub-Gaussian norm, and there exists a constant $c \geq 1$ such that for all $a \in \bbR^n$ and $p \ge 2$,
\[
\normp{\inner{a}{h}}{L_p}
\leq c \cdot 
\sqrt p\;\norm{a}_2.
\]
\end{boximportedtheorem}

The following lemma extends the $p$-th moment bound on $y_j$ given by \Cref{thm: subgaussian moment} to $d > 1$.

\begin{boxlemma}[Moment bound for sub-Gaussian chaoses]
\label{lem: chaos moment}
Let $c \ge 1$ be the constant of \Cref{thm: subgaussian moment}.  For every
order-$d$ tensor $T$, every collection of independent, isotropic, sub-Gaussian
vectors $g_1,\dots,g_d$ with constant sub-Gaussian norm, and every $p \ge 2$,
\[
\normp{\;\sum_{i_1,\dots,i_d} T_{i_1 \cdots i_d}\,(g_1)_{i_1} \cdots (g_d)_{i_d}\;}{L_p}
\le
c^{\,d}\,p^{d/2}\,\norm{T}_F.
\]
\end{boxlemma}

\begin{proof}
We induct on $d$.  When $d = 1$ the form is the linear combination
$\inner{T}{g_1}$, and \Cref{thm: subgaussian moment} gives
$\normp{\inner{T}{g_1}}{L_p} \le c \sqrt p\,\norm{T}_2$, which is the claim.

For the inductive step, group the sum according to its last index.  Writing
$S$ for the order-$d$ form and
$W_{i_d} \coloneq \sum_{i_1,\dots,i_{d-1}} T_{i_1 \cdots i_d}\,(g_1)_{i_1} \cdots
(g_{d-1})_{i_{d-1}}$ for the order-$(d-1)$ form obtained by fixing the last index,
we have $S = \sum_{i_d} (g_d)_{i_d}\,W_{i_d} = \inner{W}{g_d}$.  Crucially, the
vector $W = (W_{i_d})_{i_d}$ depends only on $g_1,\dots,g_{d-1}$, hence is
independent of $g_d$.  Conditioning on $g_1,\dots,g_{d-1}$ and applying
\Cref{thm: subgaussian moment} in the block $g_d$,
\[
\normp{S}{L_p}
\le
c\,\sqrt p\;\normp{\,\norm{W}_2\,}{L_p}.
\]
It remains to bound $\normp{\norm{W}_2}{L_p}$.  Squaring and using the triangle
inequality in $L_{p/2}$ (available since $p \ge 2$),
\[
\normp{\,\norm{W}_2\,}{L_p}^2
=
\normp{\,\textstyle\sum_{i_d} W_{i_d}^2\,}{L_{p/2}}
\le
\sum_{i_d} \normp{W_{i_d}}{L_p}^2.
\]
Each $W_{i_d}$ is a multilinear form of order $d-1$ with coefficient tensor
$T_{*, i_d}$, so the inductive hypothesis bounds
$\normp{W_{i_d}}{L_p} \le c^{\,d-1} p^{(d-1)/2}\,\norm{T_{*, i_d}}_F$.
Summing over $i_d$ and using $\sum_{i_d} \norm{T_{*, i_d}}_F^2 =
\norm{T}_F^2$ gives $\normp{\norm{W}_2}{L_p} \le c^{\,d-1} p^{(d-1)/2}\,
\norm{T}_F$.  Combining the two displays yields
$\normp{S}{L_p} \le c^{\,d} p^{d/2}\,\norm{T}_F$, completing the induction.
\end{proof}

We can now prove the tail bound.

\begin{proof}[Proof of \Cref{lem: tail}]
Applying \Cref{lem: chaos moment} to $y_j$, whose coefficient tensor $T^{(j)}$ is
just the reshaping of $U_j$ and so has Frobenius norm $\norm{U_j}_2 = 1$, gives
$\normp{y_j}{L_{2q}} \le c^{\,d}\,(2q)^{d/2}$ for every $q \ge 1$.  Combining these bounds and applying triangle inequality we obtain
\[
\normp{\norm{y}_2}{L_{2q}}^2
=
\normp{\,\textstyle\sum_{j=1}^k y_j^2\,}{L_q}
\le
\sum_{j=1}^k \normp{y_j}{L_{2q}}^2
\le
k\,c^{\,2d}\,(2q)^{d}.
\]
Markov's inequality applied to the $2q$-th moment of $\norm{y}_2$ therefore gives,
for every threshold $s > 0$,
\begin{equation*}
\label{eq:tail bound final}
\Pr{\norm{y}_2^2 \ge s\,k}
\le
\frac{\normp{\norm{y}_2}{L_{2q}}^{2q}}{(s k)^{q}}
\le
\left(\frac{c^{\,2d}\,(2q)^{d}}{s}\right)^{q}.
\end{equation*}

It remains to choose the moment order $q$.  Set $C \coloneq (2ec^2)^d=2^{O(d)}$, and fix a threshold of the
form $s = C\,t$ with $t \ge 1$.  Taking
$q =  t^{1/d}$, we have $(2q)^d \le 2^d \cdot t$ and
\[
\frac{c^{\,2d}\,(2q)^{d}}{s}
\le
\frac{(2c^2)^d\,t}{C\,t}
\le
\frac{1}{e^d}.
\]
Hence,
\[
\Pr{\norm{y}_2^2 \ge C\,k\,t}
\le
e^{-dq}
=\exp \left(-d \cdot t^{1/d}\right).
\]
\end{proof}

\subsection{Proof of \Cref{thm: main thm}}
\label{sec:final proof}

Our goal is to prove that with probability $1-\delta$ we have $\norm{\Psi^\top x}_2 = (1\pm \eps)  \|x\|_2$ for all $x \in \mathrm{Span}(U)$.
Up to constant factors in $\eps$, we can rephrase our desiderata as
\begin{equation*}
\label{eq:desiderata}
\Pr{(1-\eps)I_k \preceq U^\top \Psi \Psi^\top U \preceq (1+\eps)I_k} \geq 1-\delta.
\end{equation*}
Let $\Psi = \tfrac{1}{\sqrt m}\,[\phi_1 \mid \cdots \mid \phi_m]$, so that
$\phi_1,\dots,\phi_m$ are the columns of $\Phi$, and set
$z_i \coloneq U^\top \phi_i$.  Then
\[
U^\top \Psi \Psi^\top U = \frac1m \sum_{i=1}^m z_i z_i^\top,
\]
is the average of $m$ statistically independent rank-one matrices.
Our goal is to
show
\begin{equation}
    \label{eq:desiderata-pp}
\Pr{(1-\eps)I_k \preceq \frac1m \sum_{i=1}^m z_i z_i^\top \preceq (1+\eps)I_k} \geq 1-\delta.
\end{equation}
To prove \Cref{eq:desiderata-pp} we will use the following matrix concentration bounds.

\begin{boximportedtheorem}[Matrix Chernoff bounds \cite{tropp2012user}]
\label{thm: matrix chernoff}
Let $Y_1,\dots,Y_m \in \bbR^{k \times k}$ be independent, positive semidefinite,
and satisfy $Y_i \preceq R^2 I_k$ almost surely.  Write
$\mu_{\min}$ and $\mu_{\max}$ for the smallest and largest eigenvalues of
$\sum_i \Ex{Y_i}$.  Then for every $t \in (0,1)$,
\begin{align*}
\Pr{\lambda_{\max}\!\Bigl(\textstyle\sum_i Y_i\Bigr) \ge (1+t)\mu_{\max}}
&\le
k\,\exp\!\left(-\frac{t^2 \mu_{\max}}{3R^2}\right),
\\[2pt]
\Pr{\lambda_{\min}\!\Bigl(\textstyle\sum_i Y_i\Bigr) \le (1-t)\mu_{\min}}
&\le
k\,\exp\!\left(-\frac{t^2 \mu_{\min}}{2R^2}\right).
\end{align*}
\end{boximportedtheorem}

\paragraph{A truncation argument.} 
Unfortunately, the quality of the bound in \Cref{thm: matrix chernoff} depends on a uniform bound on the spectral radius of all $z_i z_i^\top$ for $i \in [m]$.
However, the terms $z_i z_i^\top$ might have arbitrarily large spectral radius. We obviate this issue via a truncation argument. Let
\begin{equation*}
\label{eq:def R}
R^2 \coloneq C\,k\,\left(2+ \frac{\log(m/\delta)}{d}\right)^d,
\end{equation*}
which is a tail threshold of \Cref{lem: tail} at failure probability
$\delta/m$, and accordingly $C = 2^{O(d)}$. Then, set $\wt z_i \coloneq z_i\,\ind_{\{\norm{z_i}_2 \le R\}}$ and
$Y_i \coloneq \wt z_i \wt z_i^\top$.  By construction $0 \preceq Y_i \preceq R^2
I_k$.  By the choice of $R$ and \Cref{lem: tail}, each $i\in [m]$ satisfies
$\norm{z_i}_2 > R$ with probability at most $\delta/m$, so a union bound over
$i \in [m]$ shows that with probability at least $1 - \delta$ no $z_i$ is
truncated. Thus the following observation holds.

\begin{boxobservation}
    \label{obs:truncation immaterial}
With probability at least $1-\delta$ we have $\tfrac1m \sum_i Y_i = U^\top \Psi \Psi^\top U$.    
\end{boxobservation}

Our next goal is to prove that 
\[
\Pr{(1-\eps)I_k \preceq \frac1m \sum_{i=1}^m Y_i \preceq (1+\eps)I_k} \geq 1-\delta.
\]
We will do this in two steps: first, we prove that the expectation of $\frac1m \sum_{i=1}^m Y_i$ is nearly the identity, and then we prove that $\sum_{i=1}^m Y_i$ concentrates about its expectation using \Cref{thm: matrix chernoff}.

\begin{boxlemma}
    \label{lem:expectation of Y_i}
$(1-O(\eps)) I_k \preceq \Ex{\frac1m \sum_{i=1}^m Y_i} \preceq I_k$.
\end{boxlemma}
\begin{proof}
By \Cref{lem: isotropy}, we have $\Ex{z_i z_i^\top} = I_k$, so
\[
\Ex{Y_i}
=
I_k - \Ex{z_i z_i^\top \ind_{\{\norm{z_i}_2 > R\}}},
\]
and the correction term $\Ex{z_i z_i^\top \ind_{\{\norm{z_i}_2 > R\}}}$ is positive semidefinite, hence $\Ex{Y_i} \preceq I_k$.
Writing $W \coloneq \norm{z_i}_2^2$, the operator norm of the correction term is at most $\Ex{W\,\ind_{\{W > R^2\}}}$.
Read as a tail bound, \Cref{lem: tail} states
$\Pr{W > t} \le \exp(-d(t/(Ck))^{1/d})$ for $t \ge Ck$.  Setting
$s_0 \coloneq d(R^2/Ck)^{1/d} = 2d+\log(m/\delta)$ and substituting
$s = d(t/(Ck))^{1/d}$ in the layer-cake formula for $\Ex{W\,\ind_{\{W > R^2\}}}$,
\begin{align*}
\Ex{W\,\ind_{\{W > R^2\}}}
&=
R^2\,\Pr{W > R^2} + \int_{R^2}^\infty \Pr{W > t}\,dt
\\
&\le
R^2\,e^{-s_0} + Ck\,d^{1-d}\int_{s_0}^\infty s^{d-1}e^{-s}\,ds
\\
&=
O(R^2\,e^{-s_0})
\\
&=
O\left(R^2\cdot\frac{\delta}{m}\right),
\end{align*}
where the third line holds because $s_0 \ge 2d$ and by the definition of $R$.
 
Moreover, we can verify that $m \geq 2^{O(d)} \cdot \frac{k}{\eps^2} \cdot \left(\log \left(\frac{k}{\eps \delta}\right)\right)^{d+1}$ satisfies $m = \Omega \left(R^2 \eps^{-1}\right)$, so $\Ex{W\,\ind_{\{W > R^2\}}} = O(\epsilon)$. Increasing the hidden constant in $m$ makes this bound at most $\eps/4$. Therefore, $(1-O(\eps)) I_k \preceq \Ex{Y_i} \preceq I_k$. Averaging over $i \in [m]$ yields the desiderata.
\end{proof}

\begin{boxlemma}
\label{lem:concentration of sum Y_i}
$\Pr{(1-O(\eps))I_k \preceq \tfrac1m \sum_i Y_i \preceq (1+O(\eps))I_k} \geq 1-\delta$.
\end{boxlemma}
\begin{proof}
The matrices $Y_1,\dots,Y_m$ are independent, positive
semidefinite, and $\preceq$-bounded from above by $R^2 I_k$, so \Cref{thm: matrix chernoff} applies to
their sum $\sum_i Y_i$.  
By \Cref{lem:expectation of Y_i}
we have $\mu_{\max} = \lambda_{\max}(\sum_i \Ex{Y_i}) \le m$ and
$\mu_{\min} = \lambda_{\min}(\sum_i \Ex{Y_i}) \ge m(1-O(\eps)) \ge 3m/4$.

Applying
the two tails of \Cref{thm: matrix chernoff} and taking
a union bound, the truncated average fails to satisfy
$(1-O(\eps))I_k \preceq \tfrac1m \sum_i Y_i \preceq (1+O(\eps))I_k$ with probability at
most $2k \cdot \exp(-\Omega(\eps^2 m/R^2))$.  The latter is at most $\delta$ as long as 
\begin{equation*}
\label{eq:bound on m}
m  = \Omega\left(\eps^{-2} R^2\,\log\frac{k}{\delta}\right)
\end{equation*}
which is satisfied by
$m \ge
2^{O(d)} \cdot \frac{k}{\eps^2} \cdot \left(\log \left(\frac{k}{\eps \delta}\right)\right)^{d+1}$.
\end{proof}

Finally, a union bound  over the failure probabilities of \Cref{obs:truncation immaterial} and \Cref{lem:concentration of sum Y_i} gives a combined failure probability of at most $2\delta$. When both these events happen we have the desiderata, up to constant factors in $\epsilon$ and $\delta$, which can be absorbed into $m$.

\paragraph{AI Acknowledgement.}
We acknowledge the use of AI throughout the research process. Our proofs were developed through interactive conversation with LLMs. AI was also a valuable tool for typesetting and proofreading.
The authors certify correctness of manuscript and bear responsibility for it.

\printbibliography
\end{document}